\documentclass[12pt]{iopart}

\usepackage{graphicx}% Include figure files
\usepackage{color}%colored text
\usepackage{amsthm}%theorem environment
\usepackage{amsfonts}

\theoremstyle{remark}
\newtheorem{obs}{Observation}
\newtheorem{ex}{Example}

\theoremstyle{definition}
\newtheorem{thm}{Theorem}

\begin{document}

\title{Direct evaluation of pure graph state entanglement}

\author{M Hajdu\v{s}ek$^1$ and M Murao$^{1,2}$}

\address{$^1$Department of Physics, Graduate School of Science, University of Tokyo, 7-3-1 Hongo, Bunkyo-ku, Tokyo, Japan, 113-0033}
\address{$^2$Institute for Nano Quantum Information Electronics, University of Tokyo, 4-6-1 Komaba, Meguro-ku, Tokyo, Japan, 153-8505}
\ead{michal@eve.phys.s.u-tokyo.ac.jp}

\begin{abstract}
We address the question of quantifying entanglement in pure graph states. Evaluation of multipartite entanglement measures is extremely hard for most pure quantum states. In this paper we demonstrate how solving one problem in graph theory, namely the identification of maximum independent set, allows us to evaluate three multipartite entanglement measures for pure graph states. We construct the minimal linear decomposition into product states for a large  group of pure graph states, allowing us to evaluate the Schmidt measure. Furthermore we show that computation of distance-like measures such as relative entropy of entanglement and geometric measure becomes tractable for these states by explicit construction of closest separable and closest product states respectively. We show how these separable states can be described using stabiliser formalism as well as PEPs-like construction. Finally we discuss the way in which introducing noise to the system can optimally destroy entanglement.
\end{abstract}

\maketitle

\section{\label{sec:Introduction}Introduction}

Ever since the realisation of entanglement's importance in quantum information processing a large effort has been devoted to classifying states according to their entanglement properties \cite{Horodecki_review:2009}. This has proven to be a daunting task since unlike in the bipartite case we encounter a much richer structure when characterising multipartite quantum states \cite{Bennett:2000,Dur:2000}. Even when one concentrates on a coarser picture based on separability properties the question remains formidably hard to settle \cite{Dur:2000_sep}.

Entanglement quantification in multipartite quantum states is one of the fundamental problems in quantum information theory. In this work we concentrate on three measures of genuine multipartite entanglement, namely the relative entropy of entanglement \cite{Vedral:1997}, geometric measure \cite{Shimony:1995,Barnum:2001,Wei:2003}, and the Schmidt measure \cite{Eisert:2001}. As all three measures are defined as minimisations of distances in Hilbert space or over all linear decompositions into product states they are extremely hard to compute analytically. In this paper we restrict ourselves to pure graph states.

Examples of states for which any of these measures can be computed are sparse and usually contain some form of symmetry or admit an efficient description that facilitates the evaluation. The class of symmetric states is one such example for which the relative entropy of entanglement and the geometric measure can be computed \cite{Vedral:2004,Wei:2004}. The form of the closest separable state for pure and mixed cluster states, a particular regular instance of a graph state, has been investigated in \cite{Markham:2007,Hajdusek:2010}. More general treatments of finding the closest separable state were also attempted by inverting the problem and asking what the closest entangled state is given a separable state on the boundary between entangled and separable states \cite{Ishizaka:2003,Miranowicz:2008}. However these approaches are limited to the scenario of two qubits.

Computation of relative entropy of entanglement and the geometric measure for some classes of graph states has been considered in \cite{Markham:2007}. The Schmidt measure has been found for most of 7-qubit pure graph states equivalence classes generated by local Clifford transformations in \cite{Hein:2004}. This approach was extended to 8 qubits in \cite{Cabello:2009}. A common feature among these treatments is that the amount of entanglement was found indirectly by computing the lower and upper bounds for the respective measures. In many cases the bounds coincide and therefore the exact value for entanglement can be obtained.

We extend and modify these techniques to show how the entanglement measures can be computed directly. This may seem as a daunting task at first but we demonstrate how the optimisation problem of evaluating the entanglement measures can be mapped to a well-known problem in graph theory of finding the maximum independent set, or equivalently the minimum vertex cover. Doing this allows us to approach the optimisation problem from a more graphical and intuitive perspective. Furthermore this strategy reveals a previously unrecognised connection between the closest separable state and the minimum linear decomposition of the graph state into product states.

Structure of the paper is the following. In Section~\ref{subsec:Entanglement measures} we review the three multipartite entanglement measures evaluated in this paper and known relationship between them. For the reader's convenience we also give a brief discussion of graph theory and some quantities that are central to our argument in Section \ref{subsec:Graphs}. Section~\ref{subsec:Graph states} offers a quick overview of graph states and stabiliser formalism. Finally we conclude this review section by a discussion of lower and upper bounds on the three entanglement measures in Section~\ref{subsec:Bounding}. 

In Section~\ref{sec:Evaluating} we present our main results and discuss in detail how the entanglement measures can be evaluated directly. As we are able to evaluate the three entanglement measures only in the case when the lower and upper bounds coincide we discuss the necessary and sufficient conditions when this is the case in Section~\ref{subsec:bounds_detailed}. Section~\ref{subsec:eval} presents our main results and describes how to evaluate the three entanglement measures directly. The main logic of our argument is to identify a suitable stabiliser subgroup of the original graph state stabiliser. This subgroup stabilises a subspace spanned by product states that can be used to construct the closest separable state and hence evaluate the relative entropy of entanglement. Furthermore this subspace contains the original graph state as the equal superposition of the spanning product states allowing us to find the minimal linear decomposition of the graph state vector into product states and hence to compute the Schmidt measure. Knowing the minimum linear decomposition we can easily identify the closest product states that maximise the overlap with the pure graph state making it possible to evaluate the geometric measure.

Even though the stabiliser formalism offers the most comprehensive picture when constructing the closest separable state it is limited in the sense that it only applies to pure graph states. Therefore in Section~\ref{sec:Alternative} we explore two other descriptions of the closest separable state. The first description is based on a construction similar to that of projected entangled pairs (PEPs) \cite{Verstraete:2004}. Usually PEPs are used to describe pure entangled states whereas we show how the same argument can be used to describe mixed separable states. The main motivation behind this approach is to generalise our previous techniques to states that cannot be described within the stabiliser formalism such as weighted graph states \cite{Hartmann:2007}. Our second description of the closest separable states considers how one can optimally destroy entanglement by introducing noise to the graph state. 

Finally in Section~\ref{sec:Conclusions} we summarise our work and present possible directions of future research and some open questions.

\section{\label{sec:Preliminaries}Preliminaries}

In this section we review the three entanglement measures considered in this paper and some of their relevant properties. As the derivation of majority of our result relies on concepts from graph theory we also briefly review the relevant quantities. Finally we summarise the main properties of graph states including how they can be described using the stabiliser formalism and how the bounds of entanglement measures can be calculated.

\subsection{\label{subsec:Entanglement measures}Entanglement measures}

\textit{Relative entropy of entanglement} (REE) of an entangled state $\rho$ is defined as follows \cite{Vedral:1997}
\begin{equation*}
	E_R(\rho):=\min_{\omega\in SEP}S(\rho||\omega),
\end{equation*}
where $S:=\mathrm{Tr}[\rho\log\rho-\rho\log\omega]$ is the quantum relative entropy and the minimisation is taken over all separable states $\omega$. Since we investigate only pure states in this paper the relative entropy takes a simplified form of $S=-\mathrm{Tr}[\rho\log\omega]$. We refer to the separable state that achieves the minimum of the relative entropy as the \textit{closest separable state} (CSS). Even though the relative entropy is not a true metric we can interpret $E_R(\rho)$ as the shortest distance between the entangled state $\rho$ and the set of separable states $SEP$.

\textit{Geometric measure} \cite{Shimony:1995,Barnum:2001,Wei:2003} for a pure state $|\psi\rangle$ can be defined as
\begin{equation*}
	E_G(|\psi\rangle):=\min_{|\phi\rangle\in PROD}-\log|\langle\phi|\psi\rangle|^2,
\end{equation*}
where the minimisation is taken over all product states $|\phi\rangle$. Note that in the case of geometric measure the minimisation is taken over pure states and not all mixed states like in the case of REE. Therefore it is usually easier to compute $E_G$ rather than $E_R$. We call the product state that achieves the maximum overlap with $|\psi\rangle$ the \textit{closest product state} (CPS). For general entangled states $E_G$ gives the lower bound to $E_R$ but for pure stabiliser states the measures are equal \cite{Hayashi:2006}.

Consider a pure state $|\psi\rangle\in\mathcal{H}_1\otimes\ldots\otimes\mathcal{H}_N$ of an $N$-partite system. It can be written as
\begin{equation*}
	|\psi\rangle=\sum_{i=1}^{R}\xi_i|\psi_i^1\rangle\otimes\ldots\otimes|\psi_i^N\rangle,
\end{equation*}
where $\xi_i\in\mathbb{C}$. The \textit{Schmidt measure} \cite{Eisert:2001} is then defined as
\begin{equation*}
	E_S(|\psi\rangle)=\log R_{min},
\end{equation*}
where $R_{min}$ is the minimal number of terms in the expansion of $|\psi\rangle$ over all linear decompositions into product states. Note that $E_S$ is not continuous but this does not pose a problem as the set of graph states is also discrete. 

\subsection{\label{subsec:Graphs}Graphs}

A \textit{graph} is a pair $G=(V,E)$ \cite{Diestel:2010}. Elements of $V=\{1,\ldots,N\}$ are the vertices and elements of $E\subseteq[V]^2$ are the edges connecting the vertices. Pictorially one represents graphs as a set of dots, representing the vertices, connected by lines according to $E$, representing the edges. We consider only \textit{simple graphs}, that is graphs where the vertices are not connected to themselves so $(a_i,a_i)\notin E$, and the graph contains no multiple edges between the same set of vertices. Furthermore we concentrate only on connected graphs, that is any two vertices $a_i, a_j\in V$ are connected by a path in $G$.

The \textit{neighbourhood} of a vertex $a\in V$, denoted by $N_a$, is defined as the set of all vertices that are adjacent to vertex $a$, $N_a:=\{b\in V|(a,b)\in E\}$. The \textit{degree} $d(a)$ of a vertex $a$ is in our case the size of $a$'s neighbourhood $|N_a|$. A \textit{vertex colouring} is a map $c:V\rightarrow S$ such that $c(v)\neq c(w)$ when vertices $v$ and $w$ are adjacent. The elements of $S$ are called the colours. A graph is called \textit{bipartite} if it is two-colourable. If a graph cannot be coloured by only two colours then it is called \textit{non-bipartite}.

One crucial quantity that we make use of extensively is the \textit{independent set}. An independent set is a set of vertices where no pair is adjacent. A \textit{maximum independent set} is the largest independent set for a given graph and is denoted by $\alpha(G)$. A closely related concept is that of a \textit{vertex cover}. A vertex cover is a set of vertices such that each edge of the graph is incident to at least one vertex in the vertex cover. \textit{Minimum vertex cover}, $\beta(G)$, is then naturally the smallest such set of $G$. The relationship between $\alpha(G)$ and $\beta(G)$ is that they are complements of each other which means that $\alpha(G)+\beta(G)=V$. This means that finding one set automatically gives the other. However identifying either the maximum independent set or the minimum vertex cover is a known NP-hard problem \cite{Karp:1972}.

Finally we mention a very useful transformation of a graph known as \textit{local complementation} \cite{Bouchet:1993}. Local complement (LC) of $G$ at vertex $a$, denoted by $\tau_a(G)$, is obtained by complementing the subgraph of $G$ induced by the neighbourhood $N_a$, and leaving the rest of the graph unchanged. It is equivalent to adding a fully connected graph of vertices $N_a$, denoted $G(N_a)$, to the original graph modulo 2
\begin{equation*}
	\tau_a:G\mapsto\tau_a(G):=G+G(N_a).
\end{equation*}
Applying local complementation to a vertex $a$ adds edges between its neighbours where the addition is performed modulo 2 as demonstrated in Fig.~\ref{fig:LC}.

\begin{figure}[t]
\centering
\includegraphics[scale=2.3]{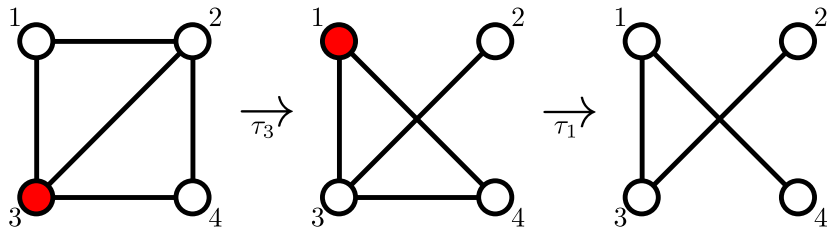}
\caption{\label{fig:LC} Successive application of local complementation first to vertex 3 and then to vertex 1. The corresponding graph states of all three of these graphs are LC-equivalent.}
\end{figure}

\subsection{\label{subsec:Graph states}Graph states}

Consider a graph $G$ and its associated \textit{graph state} $|G\rangle$ which can be prepared by placing a qubit at each vertex in the state $|+\rangle=\frac{1}{\sqrt{2}}(|0\rangle+|1\rangle)$ and applying the entangling control-Z gate $CZ_{ij}$ between all vertices $i$ and $j$ that are adjacent \cite{Briegel:2001,Raussendorf:2001}
\begin{equation}\label{eq:graph_state}
	|G\rangle=\prod_{(i,j)\in E}CZ_{ij}|+\rangle^{\otimes V}.
\end{equation}

An alternative and more efficient way of describing the graph states is using the stabiliser formalism. The graph state is the unique, common eigenvector in $(\mathbb{C}^2)^V$ to the set of independent commuting observables $\{g_i\}_{i=1}^N$ given by \cite{Hein:2006}
\begin{equation}\label{eq:stabilizer}
	g_i:=X_i\bigotimes_{j\in N_i}Z_j.
\end{equation}
The eigenvalues of the correlation operators in Eq.~(\ref{eq:stabilizer}) are $+1$ for all $i\in V$. The Abelian subgroup $\mathcal{S}$ of the Pauli group $\mathcal{P}^V$ generated by the set of all the correlation operators $\{g_i|i\in V\}$ is called the stabiliser of $|G\rangle$. We can generate a basis for $(\mathbb{C}^2)^V$ using the graph state $|G\rangle$ by defining the set of states
\begin{equation*}\label{eq:graph_basis}
	g_i|G_{k_{1}\ldots k_{i}\ldots k_{N}}\rangle=(-1)^{k_{i}}|G_{k_{1}\ldots k_{i}\ldots k_{N}}\rangle,
\end{equation*}
where $k_{1}\ldots k_{N}$ is a binary string. The original graph state given by Eq.~(\ref{eq:graph_state}) is simply $|G_{0\ldots 0}\rangle$. These states are all locally equivalent as $|G_{k_{1}\ldots k_{N}}\rangle=\prod_{i=1}^NZ_i^{k_{i}}|G_{0\ldots0}\rangle$ and therefore they all have the same amount of entanglement. The projector onto the graph state can be written in the following nice form
\begin{equation}\label{eq:graph_state_projector}
|G\rangle\langle G|=\frac{1}{2^N}\sum_{\sigma\in\mathcal{S}}\sigma,
\end{equation}
where the sum is taken over all elements of the stabiliser $\mathcal{S}$.

A graph state $|G\rangle$ corresponds uniquely to some graph $G$. However a situation may arise when two distinct graphs $G$ and $G'$ represent two graph states that are related by some unitary operation $|G'\rangle=U|G\rangle$. Their stabilisers transform accordingly $\mathcal{S}'=U\mathcal{S}U^{\dag}=\{U\sigma U^{\dag}|\sigma\in\mathcal{S}\}$. So it can easily happen that two seemingly different graphs represent graph states with the same amount of entanglement. In this paper we consider such unitary operations that permute the elements of the Pauli group and therefore map stabilisers to other stabilisers. These transformations are called \textit{local Clifford operations}, $C_1^V=\{U\in\mathbf{U}(2)^V|U\mathcal{P}^VU^{\dag}=\mathcal{P}^V\}$. So two graph states $|G\rangle$ and $|G'\rangle$ are \textit{LC-equivalent} iff they are related by some local Clifford unitary $U\in C_1^V$. In fact the action of local Clifford unitaries on graph states can be described graphically using local complementation transformations on the corresponding graph \cite{Nest:2004}. Applying the local complementation transformation to vertex $a$ of graph $G$ yields a new graph $\tau_a(G)$. The corresponding graph states $|G\rangle$ and $|\tau_a(G)\rangle$ are then LC-equivalent and are related by local Clifford operation $|\tau_a(G)\rangle=U_a^{\tau}|G\rangle$ given by
\begin{equation*}\label{eq:Clifford_unitary}
	U_a^{\tau}=\sqrt{-iX_a}\otimes\sqrt{iZ_{N_{a}}}.
\end{equation*}
Two graph states $|G\rangle$ and $|G'\rangle$ are LC-equivalent if their corresponding graphs are related by a sequence of local complementations $G'=\tau_{a_{1}}\circ\ldots\circ\tau_{a_{n}}(G)$. Furthermore the local Clifford unitary relating two equivalent graph states can be found efficiently \cite{Nest:2004b}.

\textcolor{red}{}

\subsection{\label{subsec:Bounding}Bounding the entanglement}

Finding a candidate separable state that can be used in quantifying the entanglement is often a very hard task. However proving that this state achieves the minimum value of relative entropy or the maximum overlap is equally difficult. Fortunately the situation for many classes of graph states is such that lower and upper bounds on the entanglement measures coincide therefore the exact value of entanglement is known. This makes direct evaluation of entanglement measures easier because it is enough to find separable state $\omega$ and product state $|\phi\rangle$ that yield this value for $E_R$ and $E_G$ respectively. In the case of the Schmidt measure $E_S$, knowing its value tells us the minimum number of  terms the linear decomposition of $|G\rangle$ into product states needs to contain.

In \cite{Markham:2007} the authors showed how the upper and lower bounds for relative entropy of entanglement and geometric measure can be found. The upper bound is obtained by considering perfect LOCC discrimination of a subset of the complete orthogonal graph state basis $\{|G_{k_{1}\ldots k_{N}}\rangle\}$. A lower bound on the number of these states that can be discriminated is given by maximising the number of stabiliser generators $\{g_i\}$ that can be determined in a single setting of LOCC measurements. This bound can be achieved by identifying the maximum independent set $\alpha(G)$ of the graph $G$ and measuring $X$ on these qubits and $Z$ on their neighbours. This can in fact be used to find an upper bound on the relative entropy of entanglement and the geometric measure $N-|\alpha(G)|\geq E_R=E_G$ \cite{Markham:2007}.

Lower bound for the entanglement can be obtained in the usual way of relaxing the condition of full separability and maximising the entanglement $E_{bi}$ between all the bipartitions of the graph. This can be viewed as creating the maximum number of Bell pairs $m_p$ between the bipartitions by $CZ$ and local Clifford operations applied within the bipartitions. Summarizing the lower and upper bounds for the relative entropy of entanglement and the geometric measure are \cite{Markham:2007}
\begin{equation*}\label{eq:bounds_REE_GM}
	N-|\alpha(G)|\geq E_R=E_G\geq E_{bi}=m_p.
\end{equation*}
Classes of graph states for which these bounds are equal include $d$-dimensional cluster states, GHZ states and ring states with even number of qubits.

The bounds for Schmidt measure can be found in similar fashion \cite{Hein:2004,Hein:2006}. The lower bound is given by the maximal Schmidt rank $\textrm{SR}_{\max}$ which for graph states is equal to the entropy of entanglement $\textrm{SR}_{\max}=-\textrm{Tr}[\rho^A\log\rho^A]$ where $A\subseteq V$ is a subset of the vertices. The upper bound is given by the minimal number of local Pauli measurements needed to completely disentangle the graph state. This quantity is referred to as the \textit{Pauli persistency}, PP($G$), and is bounded from above by the size of the minimum vertex cover $|\beta(G)|$ \cite{Briegel:2001}. So the Schmidt measure is bounded as follows
\begin{equation*}\label{eq:bounds_Schmidt}
	|\beta(G)|\geq\textrm{PP}(G)\geq E_S\geq \textrm{SR}_{\max}.
\end{equation*}
States up to 7 and 8 qubits for which these bounds are equal are categorised in \cite{Hein:2006} and \cite{Cabello:2009} respectively.

\section{\label{sec:Evaluating}Evaluating entanglement measures}

Now we have all the necessary tools needed to construct the minimal linear decomposition of the graph state into product states and hence the closest separable state $\omega$ and the closest product state $|\phi\rangle$. But before we do this it is necessary to mention that in many cases the lower and upper bounds coincide. In this section we present an algorithm for constructing the closest separable states that saturate the upper bound. Therefore when the bounds coincide it is trivial to show that a candidate separable state is truly the closest one. However if the bounds are not equal then it is not possible to conclusively state that the obtained separable state is the closest one. It is therefore crucial to understand when the bounds are equal and when they are not.

\subsection{\label{subsec:bounds_detailed}Equality of lower and upper bounds}

To understand when the lower and upper bounds are equal we first demonstrate how the lower bound for all three entanglement measures can be found in terms of a single quantity from graph theory.

A \textit{matching} $M$ of a graph $G$ is a set of independent edges. Two edges are independent when they do not have any common vertices. A subset of vertices $U\subseteq V$ is called matched if every vertex $a\in U$ is incident to an edge in $M$. If the entire vertex set $V$ is matched by $M$ then it is called a \textit{perfect matching}.

The notion of matching is very useful because the cardinality of maximum matching $|M_{max}|$ is equal to the maximum number of Bell pairs that can be created for a bipartition of $G$ as well as the maximal Schmidt rank $\textrm{SR}_{\max}$.

\begin{figure}[t]
	\centering
	\includegraphics[scale=2.3]{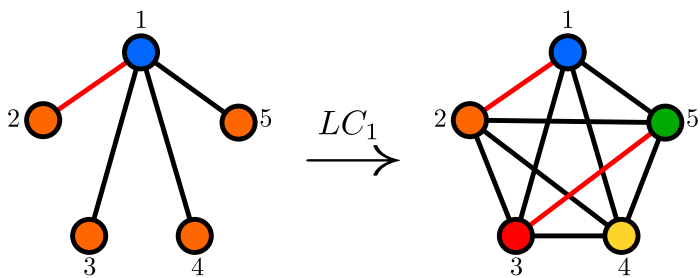}
	\caption{\label{fig:fig_2} Star graph and complete graph are related by a single local complementation operation therefore their corresponding graph states are LC-equivalent and in this case  have entanglement of 1. Care needs to be taken when computing the lower and upper bounds as the entire LC orbit of a graph, not just the graph itself, needs to be considered. Red edges depict the maximum matching.}
\end{figure}
At this point we would like to stress one crucial point that has been implicitly covered in Section~\ref{subsec:Graph states}. We have explained that two graph states $|G\rangle$ and $|G'\rangle$ which are related by a local Clifford transformation have underlying graphs $G$ and $G'$ related by a series of local complementation operations. However the properties of the two graphs may be very different. In particular their maximum matching and minimum vertex cover may have different cardinalities. One such example is the GHZ state pictured in Fig.~\ref{fig:fig_2}. This graph state corresponds to a star graph which can be transformed to a complete graph using local complementation. So the graph states corresponding to a star graph and a complete graph are LC-equivalent and therefore have the same entanglement of 1. However for the star graph we have $|M_{max}|=1$ and $|\beta(G)|=1$ whereas the complete graph has $|M_{max}|=\left\lfloor\frac{N}{2}\right\rfloor$ and $|\beta(G)|=N-1$.

To overcome this apparent contradiction it is necessary to consider the entire LC-equivalency class of a graph $G$ and compute the size of $M_{max}$ and $\beta(G)$ for all graphs in the equivalency class to find the minimum. Such canonical forms for all LC-equivalency classes have been identified up to 7 qubits in \cite{Hein:2004,Hein:2006}, up to 8 qubits in \cite{Cabello:2009} and up to 12 qubits in \cite{Cabello:2011}. For a graph state up to 12 qubits it is enough to find which canonical form it is equivalent to. However for a general graph state of $N$ qubits one has to first generate the entire LC orbit of the underlying graph. Therefore when computing the lower and upper bounds for the entanglement of a graph state it is not enough to consider only the graph itself. The entire LC orbit needs to be taken into account. In the rest of this paper we will assume that the canonical form of the underlying graph ith smallest $M_{max}$ and $\beta(G)$ across the LC orbit is known.

Consider a graph $G$ with a maximum matching $M_{max}$. Select either endvertex from each independent edge in $M_{max}$ to form one partition $A$. The other partition is then given by all the other vertices in $V$, that is $B=V-A$. Apply control-Z gates locally within the partitions to remove any edges that do not cross the bipartition boundary. As a result of this only edges that contribute towards entanglement across the bipartition are kept. In certain cases this transformation leaves all the vertices in the partitions with either degree 0 or degree 1. This corresponds to the case of Bell pairs being created across the bipartition. However in the general case some vertices will have degree higher than 1. Then we need to apply a series of local complementation transformations along with further control-Z gates to either decrease their degree to 1 or completely disconnect them from the rest of the graph. Crucial observation here is that none of the above transformations delete an edge from the maximum matching $M_{max}$. Therefore this procedure produces $|M_{max}|$ Bell pairs across the bipartition. This method of transforming the graph state has been introduced in \cite{Markham:2007} and demonstrated for a selection of graph states with underlying bipartite graphs. Fig.~\ref{fig:fig_3} displays the procedure for a particular example of a non-bipartite graph.
\begin{figure}[t]
	\centering
	\includegraphics[scale=1.8]{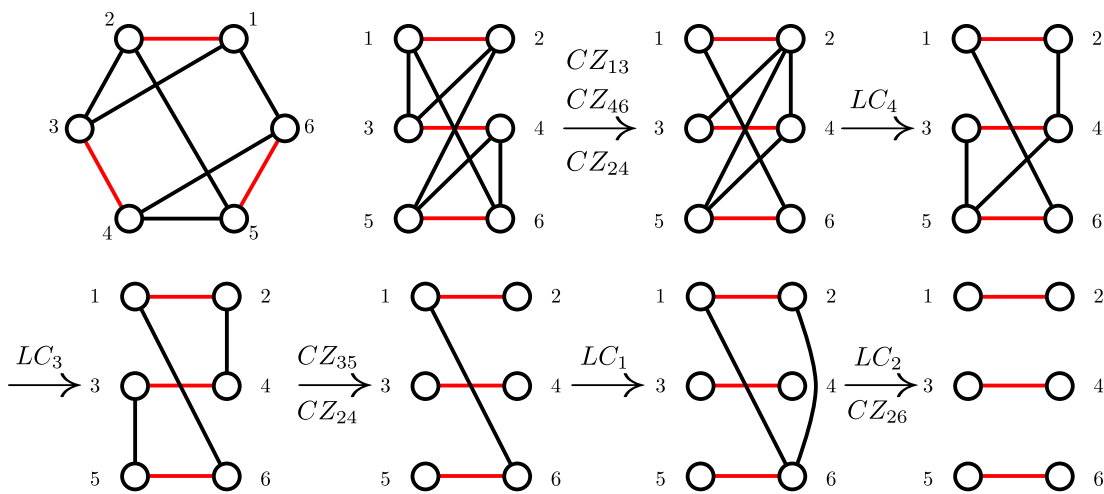}
	\caption{\label{fig:fig_3} Non-bipartite graph (in \cite{Hein:2006} denoted as No.~19). The maximum matching $M_{max}=\{(1,2),(3,4),(5,6)\}$ is represented by red colour. The bipartition is given by the sets $A=\{1,3,5\}$ and $B=\{2,4,6\}$. Successive application of local control-Z gates within the partitions and local complementation results in a disconnected graph representing 3 Bell pairs.}
\end{figure}

Having established that the lower bound for all three measures is given by the size of the maximum matching $|M_{max}|$ and the upper bound by the cardinality of the minimum vertex cover $|\beta(G)|$ we will now investigate when these bounds are equal and when not. To do this we will partition the graph states into the set of bipartite and non-bipartite states, and consider them separately.

\textit{Bipartite graph states}. For all bipartite graph states the lower and upper bounds coincide. This can be seen immediately using K\"{o}nig's Theorem \cite{Diestel:2010} which states that for bipartite graphs the sizes of maximum matching and of minimum vertex cover are equal, $|M_{max}|=|\beta(G)|$.

\textit{Non-bipartite graph states.} The situation for non-bipartite graph states is more complicated as K\"{o}nig's Theorem does not hold anymore and in general $|M_{max}|\neq|\beta(G)|$. However there are still numerous cases when the two bounds are equal. In fact the crucial quantity that determines whether the bounds are equal is the size of the maximum independent set $|\alpha(G)|$:
\begin{enumerate}
	\item If $|\alpha(G)|<\frac{N}{2}$ then $|M_{max}|\neq|\beta(G)|$.
	\item If $|\alpha(G)|>\frac{N}{2}$ then $|M_{max}|=|\beta(G)|$.
	\item If $|\alpha(G)|=\frac{N}{2}$ then we have the following two scenarios:
	\begin{enumerate}
		\item If $M_{max}$ is perfect then $|M_{max}|=|\beta(G)|$.
		\item If $M_{max}$ is not perfect then $|M_{max}|\neq|\beta(G)|$.
	\end{enumerate}
\end{enumerate}

Statement~1 can be proven easily. For a graph state of $N$ qubits the maximum number of Bell pairs that can be created is $|M_{max}|\leq\left\lfloor \frac{N}{2}\right\rfloor$. However the upper bound for entanglement is given by $N-|\alpha(G)|>\frac{N}{2}$ from our assumption. Therefore the bounds are not equal, $|M_{max}|\neq|\beta(G)|$.

Statement~2 becomes clear when one considers a bipartition of the graph $G$ with partition $A$ given by the vertices in the maximum independent set $A:=\{v|v\in\alpha(G)\}$ and partition $B$ given by all the other vertices, that is the minimum vertex cover $B:=\{v|v\in\beta(G)\}$. It is clear that all vertices $v\in A$ are not connected. Now consider the reduced neighbourhood of a vertex $a\in B$ defined as $N^{red}_a:=\{v\in A|(a,v)\in E\}$. So $N^{red}_a$ is just the usual neighbourhood with vertices in $B$ removed. When looking for the maximum matching of graph $G$ we can consider each $N^{red}_a$ separately. The crucial point to note here is that since all $v\in N^{red}_a$ are incident to a single vertex $a\in B$ we can only pick one edge that contributes to the maximum matching. The fact that $|M_{max}|=|\beta(G)|$ follows straightaway. Including edges between vertices in $B$ in the matching can only decrease the matching's cardinality so they have been omitted from the argument. This is illustrated in Fig.~\ref{fig:fig_4}.
\begin{figure}[t]
	\centering
	\includegraphics[scale=2.3]{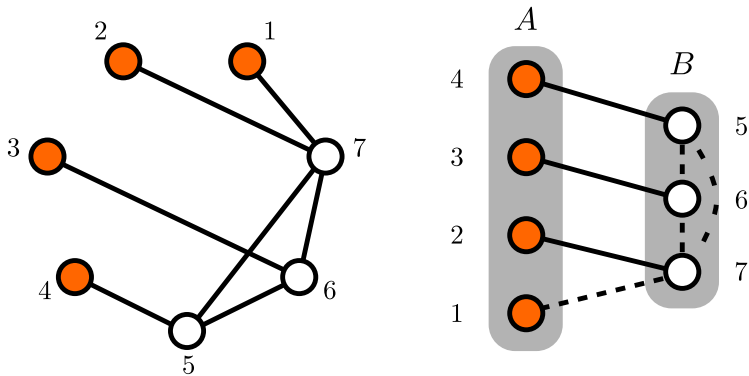}
	\caption{\label{fig:fig_4} Graph $G$ with $|\alpha(G)|>\frac{N}{2}$. The maximum independent set is represented by orange vetices. The reduced neighborhoods are $N^{red}_5=\{4\}$, $N^{red}_6=\{3\}$ and $N^{red}_7=\{1,2\}$. The maximum matching is depicted by edges with solid lines across the bipartition $A-B$.}
\end{figure}

Statement~3 becomes clear after one realises what it means to have a perfect matching. The upper bound for entanglement in both cases (a) and (b) is given by $|\beta(G)|=N-|\alpha(G)|=\frac{N}{2}$. If the matching $M_{max}$ is perfect then each vertex in the graph is matched. Since a matching is a set of independent edges this implies that $|M_{max}|=\frac{N}{2}$. Therefore the entanglement bounds coincide, $|M_{max}|=|\beta(G)|$. On the other hand if $M_{max}$ is not a perfect matching then there exists at least one vertex which is not matched by $M_{max}$. Hence the number of vertices that can are matched by $M_{max}$ is upper bounded by $N-1$. We know that $N$ is even so the size of the maximum matching is upper bounded by $|M_{max}|\leq\frac{N-2}{2}$. This shows that if $|\alpha(G)|=\frac{N}{2}$ and the maximum matching is not perfect then the bounds for entanglement do not coincide.

\subsection{\label{subsec:eval}Entanglement evaluation}

As mentioned in Section~\ref{subsec:Graph states} the stabiliser $\mathcal{S}$ of an $N$-qubit graph state generated by $N$ generators $g_i$ given by Eq.~(\ref{eq:stabilizer}) stabilises a unique entangled state. Consider a scenario where we discard a set of $k$ generators. The new stabiliser is given by $\mathcal{S}_{N-k}=\langle g_1,\ldots,g_{N-k}\rangle$ where we have relabeled the remaining generators from $\mathcal{S}$ for convenience. Because we no longer have a full set of $N$ generators, $\mathcal{S}_{N-k}$ does not stabilise a unique state. Rather it stabilises a set of states $\{|\psi_1\rangle,\ldots,|\psi_D\rangle\}$ that span a $D$-dimensional subspace. The dimensionality of this subspace depends on the particular structure of the generators $g_i\in\mathcal{S}_{N-k}$ and the states $\{|\psi_i\rangle\}$ may be entangled or product states. We are interested in the case when $\mathcal{S}_{N-k}$ stabilises a set of product states.

\begin{obs}\label{obs:obs1}
	(Stabilised entangled states): $\mathcal{S}_{N-k}$ stabilises entangled states if it contains at least two generators $g_a,g_b\in\mathcal{S}_{N-k}$ where $g_b$ contains $Z_a$.
\end{obs}

Such a scenario happens when the generators $g_a$ and $g_b$ correspond to adjacent qubits, that is $(a,b)\in E$. Generator $g_a$ stabilises some subspace $\mathcal{H}_a$ spanned by states $\{|\psi^{(a)}_i\rangle\}$ and $g_b$ stabilises a subspace $\mathcal{H}_b$ spanned by states $\{|\psi^{(b)}_j\rangle\}$. Since qubits $a$ and $b$ are adjacent, the action of $g_a$ on the states $\{|\psi^{(b)}_j\rangle\}$ is to permute them and the same is true for $g_b$ acting on $\{|\psi^{(a)}_i\rangle\}$. In order for $g_a$ to stabilise the states $\{|\psi^{(b)}_j\rangle\}$ we have to take superpositions of these states according to how they are permuted by $g_a$. This finally yields states that are stabilised by $\mathcal{S}_{N-k}$, however due to the superpositions the stabilised states are entangled. We illustrate this in the following example.

\begin{ex}\label{ex:ex1}
	Consider a three-qubit open linear graph state given by the stabiliser
	\begin{equation*}\label{eq:3_qubit_example}
		\mathcal{S}=\langle \underbrace{XZI}_{g_{1}}, \underbrace{ZXZ}_{g_{2}}, \underbrace{IZX}_{g_{3}}\rangle.
	\end{equation*}
	Say we discard the third generator so the new stabiliser is given by $\mathcal{S}_{2}=\langle XZI, ZXZ\rangle$. Each generator stabilises a set of states
	\begin{eqnarray*}
		g_1 & = & XZI:\{|+0.\rangle,|-1.\rangle\}, \\
		g_2 & = & ZXZ:\{|0+0\rangle,|1+1\rangle,|0-1\rangle,|1-0\rangle\},
	\end{eqnarray*}
	where the states of the form $|+0.\rangle$ mean that the generator does not fix the last qubit. We can quickly check that the action of $g_1$ on the states stabilised by $g_2$ is to 	permute them, that is $g_1|0+0\rangle=|1-0\rangle$ and $g_1|1+1\rangle=|0-1\rangle$. Requiring that $g_1$ stabilises all the above states we are forced to take superpositions of the 	states which finally yields the states that are stabilised by $\mathcal{S}_{2}$ to be
	\begin{equation*}\label{eq:stabilized_states_a_example}
		\{\frac{1}{\sqrt{2}}(|0+\rangle+|1-\rangle)\otimes|0\rangle,\frac{1}{\sqrt{2}}(|0-\rangle+|1+\rangle)\otimes|1\rangle\}.
	\end{equation*}
	The stabilised subspace is two-dimensional and because of the form of the generators of $\mathcal{S}_2$ it is spanned by entangled states. On the other hand we can quickly check that 	discarding $g_2$ from the original $\mathcal{S}$ produces a stabilised subspace spanned by product states $\{|+0+\rangle,|-1-\rangle\}$.
\end{ex}

The above observation tells us exactly when $\mathcal{S}_{N-k}$ stabilises a set of product states. The desired structure of the generators $g_i\in\mathcal{S}_{N-k}$ is achieved when the generators correspond to non-adjacent qubits. As a vertex colouring of the underlying graph $G$ partitions the set of qubits into subsets of independent qubits, the possible ways of discarding generators from the original stabiliser $\mathcal{S}$ are given by all the possible vertex colourings that $G$ admits. More specifically if the generators corresponding to qubits of the same colour are the only generators not discarded then $\mathcal{S}_{N-k}$ stabilises a set of product states. Note that the converse is not true as it is possible for $\mathcal{S}_{N-k}$ to contain generators corresponding to different colours and still stabilise a set of product states. This is illustrated in Fig.~\ref{fig:fig_5}.
\begin{figure}[t]
	\centering
	\includegraphics[scale=2.3]{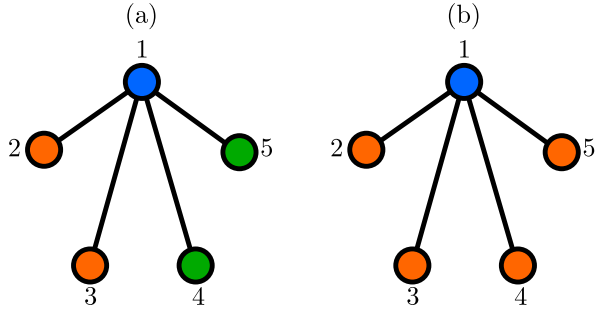}
	\caption{\label{fig:fig_5} Two valid colourings for a 5-qubit GHZ state. (a)~Both $\mathcal{S}_2=\langle g_2,g_3\rangle$ and $\mathcal{S}_2=\langle g_4,g_5\rangle$ stabilise a set of product states. Keeping certain generators corresponding to qubits with different colours can also stabilise a set of product states. Such a stabiliser is $\mathcal{S}_3=\langle g_2,g_3,g_4\rangle$. (b)~Maximum independent set, $\alpha(G)=\{2,3,4,5\}$, identifying the largest set of generators that stabilise a set of product states.}
\end{figure}

Knowing how to discard generators from $\mathcal{S}$ to produce a product basis we can focus on the important question of doing this optimally. Ideally we would like to discard as few generators as possible. In light of the above discussion it is straightforward to see that the generators that need to be kept correspond to qubits in the maximum independent set $\alpha(G)$, or equivalently, the generators that need discarding correspond to the minimum vertex cover $\beta(G)$. Similar approach has been used to compute upper bounds for the three entanglement measures in \cite{Markham:2007,Hein:2004,Hein:2006} though our logic of deriving this result is complementary to the approach in these references.

We will now show how the entanglement measures can be evaluated directly. We refer to the stabiliser of generators corresponding to qubits in the maximum independent set $\alpha(G)$ as $\mathcal{S}_{\alpha}$, the stabilised subspace as $\mathcal{H}_{\alpha}$ and product states spanning this subspace as $\{|\psi^{\alpha}_i\rangle\}$.

\begin{thm}\label{thm:Schmidt}
(Minimal linear decomposition into product states): Given a graph state $|G\rangle$, its minimal linear decomposition into product states is a superposition of states $\{|\psi^{\alpha}_i\rangle\}$
\begin{equation}\label{eq:schmidt_decomposition}
	|G\rangle=\frac{1}{\sqrt{D_{\alpha}}}\sum_{i=1}^{D_{\alpha}}(-1)^{f_i(\mathcal{S})}|\psi^{\alpha}_i\rangle,
\end{equation}
where $f_i(\mathcal{S})$ is a binary-valued function and its value depends on the action of the original stabiliser $\mathcal{S}$ on the states $|\psi^{\alpha}_i\rangle$. This ensures that the form of $|G\rangle$ in Eq.~(\ref{eq:schmidt_decomposition}) is stabilised by the whole $\mathcal{S}$. $D_{\alpha}$ is the dimension of the subspace $\mathcal{H}_{\alpha}$ and depends on the size of the minimum vertex cover as $D_{\alpha}=2^{|\beta(G)|}$.
\end{thm}

\begin{proof}
	The proof consists of two parts. First part shows that the decomposition in Eq.~(\ref{eq:schmidt_decomposition}) actually gives the graph state $|G\rangle$. The second part shows that the decomposition achieves the minimum bound of the Schmidt measure $E_S(|G\rangle)$.
	
	 From Observation~(\ref{obs:obs1}) we can be sure that the states $|\psi^{\alpha}_i\rangle$ are product states. Furthermore since these states are stabilised by $\mathcal{S}_{\alpha}$ we have $g_j|\psi^{\alpha}_i\rangle=|\psi^{\alpha}_i\rangle$ for all $j\in\mathcal{S}_{\alpha}$ and $i\in\{1,\ldots,D_{\alpha}\}$. All we have to determine is the action of generators in the minimum vertex cover $\beta(G)$ on the states $|\psi^{\alpha}_i\rangle$. Consider the action of one of the generators in the minimum vertex set $g_k$, $k\in\beta(G)$. Due to the structure of the correlation operators this permutes the elements in the stabilised set $g_k|\psi^{\alpha}_j\rangle=|\psi^{\alpha}_i\rangle$. If the qubit $k$ is adjacent to another qubit in the minimum vertex cover then the action of $g_k$ may be to also introduce a negative sign, $g_k|\psi^{\alpha}_j\rangle=-|\psi^{\alpha}_i\rangle$. By taking superpositions of states $|\psi^{\alpha}_i\rangle$ that are permuted with each other and including the negative amplitudes we can construct a new set of states $|\psi'_i\rangle$ that are stabilised by the generators $g_i\in\alpha(G)$ as well as the new generator $g_k\in\beta(G)$. The new states will be of the form
	 \begin{equation*}\label{eq:proof1_superposition}
	 	|\psi'_i\rangle=|\psi^{\alpha}_i\rangle+g_k|\psi^{\alpha}_i\rangle.
	 \end{equation*}
Repeating this process with a new generator $g_l$ where $l\in\beta(G)$ acting on the states $|\psi'_i\rangle$ we obtain a new set of states stabilised by $\mathcal{S}_{\alpha}$ as well as $g_k,g_l$. Extending this procedure to all of the generators in $\beta(G)$ we finally arrive at the state in Eq.~(\ref{eq:schmidt_decomposition}).

Showing that this decomposition is also minimal becomes now trivial for graph states for which $|M_{max}|=|\beta(G)|$ as $\log D_{\alpha}$ reaches the lower bound for the Schmidt measure known from \cite{Hein:2004,Hein:2006}. Therefore the entanglement as evaluated by the Schmidt measure is given by
\begin{equation*}\label{eq:proof_schmidt}
	E_S(|G\rangle)=\log D_{\alpha},
\end{equation*}
which concludes he proof.
\end{proof}

\begin{obs}\label{obs:obs2}
	Bipartite graph states for which the maximum independent set $\alpha(G)$ is also given by a 2-colouring of $G$ can be written as $|G\rangle=\frac{1}{\sqrt{D_{\alpha}}}\sum_{i=1}^{D_{\alpha}}|\psi^{\alpha}_i\rangle$. Because any qubit in the minimum vertex cover $\beta(G)$ has all its neighbours in the maximum independent set $\alpha(G)$, the action of any generator $g_k\in\mathcal{S}_{\beta}$ is to only permute the states stabilised by $\mathcal{S}_{\alpha}$, that is $g_k|\psi^{\alpha}_j\rangle=|\psi^{\alpha}_i\rangle$ for $i,j\in\alpha(G)$. Therefore all the amplitudes in the linear decomposition of Eq.~(\ref{eq:schmidt_decomposition}) are positive.
\end{obs}

Using the decomposition in Eq.~(\ref{eq:schmidt_decomposition}) we can now compute the other two measures.

\begin{thm}\label{thm:REE}
	(Relative entropy of entanglement): The closest separable state $\omega$ to a given graph state $|G\rangle$ is given by
	\begin{equation}\label{eq:CSS}
		\omega=\frac{1}{D_{\alpha}}\sum_{i=1}^{D_{\alpha}}|\psi^{\alpha}_i\rangle\langle\psi^{\alpha}_i|.
	\end{equation}
	So the closest separable state is given by equal mixture of the states stabilised by $\mathcal{S}_{\alpha}$.
\end{thm}

\begin{proof}
	States of the form given by Eq.~(\ref{eq:CSS}) are clearly separable because they are a mixture of product states. Computing the relative entropy between $\rho=|G\rangle\langle G|$ and $\omega$ we get
	\begin{eqnarray*}\label{eq:CSS_proof}
		S(\rho||\omega) & = & -\textrm{Tr}[\rho\log\omega] \nonumber \\
		& = & -\frac{1}{D_{\alpha}}\log\frac{1}{D_{\alpha}}\sum_{ijk}^{D_{\alpha}}(-1)^{f_{i}(\mathcal{S})+f_{j}(\mathcal{S})}\delta_{jk}\delta_{ki} \nonumber \\
		& = & \log D_{\alpha}.
	\end{eqnarray*}
	For graph states which have $|M_{max}|=|\beta(G)|$ this is the lower bound found in \cite{Markham:2007}. Therefore $\omega$ is the closest separable state and the relative entropy of entanglement is given by
	\begin{equation*}\label{eq:REE}
		E_R(|G\rangle)=\log D_{\alpha}.
	\end{equation*}
	This concludes the proof.
\end{proof}

The closest separable state $\omega$ can be expressed in a similar form to Eq.~(\ref{eq:graph_state_projector}) as the sum of all the elements in $\mathcal{S}_{\alpha}$.
\begin{obs}
	We can write the closest separable state as
	\begin{equation}\label{eq:CSS_sum_rep}
		\omega=\frac{1}{2^{N}}\sum_{\sigma\in\mathcal{S}_{\alpha}}\sigma.
	\end{equation}
	This can be most easily seen by using Eq.~(\ref{eq:CSS}) and noting that $\langle\psi_j^{\alpha}|(\frac{1}{D_{\alpha}}\sum_{i=1}^{D_{\alpha}}|\psi^{\alpha}_i\rangle\langle\psi^{\alpha}_i|)|\psi^{\alpha}_k\rangle=\frac{1}{D_{\alpha}}\delta_{jk}$. Using Eq.~(\ref{eq:CSS_sum_rep}) and the fact that $\sigma|\psi^{\alpha}_i\rangle=|\psi^{\alpha}_i\rangle$ for all $\sigma\in\mathcal{S}_{\alpha}$ we can obtain
	\begin{eqnarray*}
		\langle\psi^{\alpha}_j|\frac{1}{2^N}\sum_{\sigma\in\mathcal{S}_{\alpha}}\sigma|\psi^{\alpha}_k\rangle & = & \frac{2^{|\alpha(G)|}}{2^N}\langle\psi^{\alpha}_j|\psi^{\alpha}_k\rangle\\
		& = & \frac{1}{2^{|\beta(G)|}}\delta_{jk},
	\end{eqnarray*}
	where we used $|\alpha(G)|+|\beta(G)|=N$ and that the cardinality of $\mathcal{S}_{\alpha}$ is $2^{|\alpha(G)|}$. This establishes equivalence between the two forms of $\omega$ in Eq.~(\ref{eq:CSS}) and Eq.~(\ref{eq:CSS_sum_rep}).
\end{obs}

\begin{thm}\label{thm:GM}
	(Geometric measure): The closest product state $|\phi\rangle$ is given by any state from the stabilised set $\{|\psi^{\alpha}_i\rangle\}$,
	\begin{equation}\label{eq:CPS}
		|\phi\rangle=|\psi^{\alpha}_i\rangle,\qquad\forall i\in\{1,\ldots,D_{\alpha}\}.
	\end{equation} 
\end{thm}

\begin{proof}
	Using the minimum linear decomposition into product states in Eq.~(\ref{eq:schmidt_decomposition}) and substituting $|\phi\rangle=|\psi^{\alpha}_i\rangle$ we find that the geometric measure is
	\begin{eqnarray*}\label{eq:GM_proof}
		E_G(|G\rangle) & = & -\log|\langle\psi^{\alpha}_i|\frac{1}{\sqrt{D_{\alpha}}}\sum_{j=1}^{D_{\alpha}}(-1)^{f_{j}(\mathcal{S})}|\psi^{\alpha}_j\rangle|^2 \nonumber \\
		& = & \log D_{\alpha}.
	\end{eqnarray*}
	This is the lower bound found in \cite{Markham:2007} and therefore concludes the proof.
\end{proof}

We demonstrate how to quantify the entanglement measures in the following example of a 6-qubit graph state.

\begin{ex}\label{ex:ex2}
	Consider a graph state $|G\rangle$ of 6 qubits with an underlying graph $G$ pictured in Fig.~\ref{fig:fig_6}.
	\begin{figure}[t]
	\centering
	\includegraphics[scale=2.3]{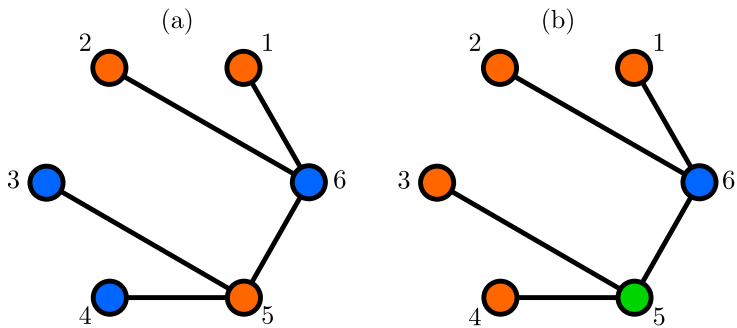}
	\caption{\label{fig:fig_6} Graph $G$ considered in Ex.~(\ref{ex:ex2}). (a)~Graph $G$ is bipartite since a 2-colouring can be found. (b)~However the maximum independent set $\alpha(G)$ is identified using a 3-colouring.}
	\end{figure}
The graph state $|G\rangle$ is stabilised by
\begin{eqnarray*}\label{eq:example2_stabilizer}
	\mathcal{S} & = & \langle \underbrace{XIIIIZ}_{g_1},\underbrace{IXIIIZ}_{g_2},\underbrace{IIXIZI}_{g_3}, \nonumber \\
	&& \underbrace{IIIXZI}_{g_4},\underbrace{IIZZXZ}_{g_5},\underbrace{ZZIIZX}_{g_6}\rangle.
\end{eqnarray*}
The maximum independent set can be quickly obtained by a 3-colouring and is $\alpha(G)=\{1,2,3,4\}$. Therefore the corresponding stabiliser is $\mathcal{S}_{\alpha}=\langle g_1,g_2,g_3,g_4\rangle$. $|\alpha(G)|>3$ which means that the lower and upper bound coincide. The stabilised states given by $\mathcal{S}_{\alpha}$ are
\begin{equation*}\label{eq:stabilized_states}
	\{|\psi^{\alpha}_i\rangle\}=\{|++++00\rangle,|--++01\rangle,|++--10\rangle,|----11\rangle\}.
\end{equation*}
The action of the stabiliser $\mathcal{S}_{\beta}$ associated with the minimum vertex cover is the following
\begin{eqnarray*}
	g_5|\psi^{\alpha}_1\rangle & = & |\psi^{\alpha}_3\rangle \qquad\textrm{and}\qquad g_5|\psi^{\alpha}_2\rangle=-|\psi^{\alpha}_4\rangle, \\
	g_6|\psi^{\alpha}_1\rangle & = & |\psi^{\alpha}_2\rangle \qquad\textrm{and}\qquad g_6|\psi^{\alpha}_3\rangle=-|\psi^{\alpha}_4\rangle,
\end{eqnarray*}
which means that the graph state can be written in the following form
\begin{equation*}
	|G\rangle=\frac{1}{2}[|\psi^{\alpha}_1\rangle+|\psi^{\alpha}_2\rangle+|\psi^{\alpha}_3\rangle-|\psi^{\alpha}_4\rangle].
\end{equation*}
Using the closest separable state $\omega$ in Eq.~(\ref{eq:CSS}) and the closest product state in Eq.~(\ref{eq:CPS}) we can finally show that for this graph state $|G\rangle$
\begin{equation*}
	E_S(|G\rangle)=E_R(|G\rangle)=E_G(|G\rangle)=2.
\end{equation*}
\end{ex}

Knowing how to find the closest separable state $\omega$ to a given graph state $|G\rangle$, it is now possible to determine the form of the closest separable states to all of the graph states in the orbit generated by local Clifford operations. Two graph states are LC-equivalent, $|G'\rangle=U^{LC}|G\rangle$, if they are related by a local unitary $U^{LC}$. Using Eq.~(\ref{eq:schmidt_decomposition}), $|G'\rangle$ can be expressed as a superposition of states $U^{LC}|\psi^{\alpha}_i\rangle$. Applying the unitary we see that the closest separable state has the form $\omega'=\frac{1}{D_{\alpha}}\sum U^{LC}|\psi^{\alpha}_i\rangle\langle\psi^{\alpha}_i|U^{LC\dagger}$.

\subsection{\label{subsec:maximal} Maximal entanglement of graph states}

An immediate consequence of this approach to analysing entanglement is that we can identify which graph states are maximally entangled in many cases. By maximally entangled we mean with respect to the three measures we are considering. For any graph state that is maximally entangled it must be true that $E_R=E_G\geq|M_{max}|$ and $E_S\geq|M_{max}|$. For bipartite graph states these measures cannot be larger than this value due to K\"onig's Theorem \cite{Diestel:2010}. Therefore any bipartite graph state $|G\rangle$ whose underlying graph $G$ has the property that $|M_{max}|=\left\lfloor\frac{N}{2}\right\rfloor$ is maximally entangled. Examples of such states include linear graph states with open boundaries, ring states with even $N$ and cluster states in $d$ dimensions. This holds true also for non-bipartite graph states for which the bounds are equal. An example of such a state is pictured in Fig.~\ref{fig:fig_4}.

The situation is quite different for the case of graph states with unequal entanglement bounds. Again any graph state that is maximally entangled must have $|M_{max}|=\left\lfloor\frac{N}{2}\right\rfloor$. However the true value of entanglement is now unclear. For general graph states it is not even possible to calculate this bound. Therefore we limit ourselves to various types of regular two dimensional lattices for the rest of this section.

\begin{figure}[t]
	\centering
	\includegraphics[scale=1.1]{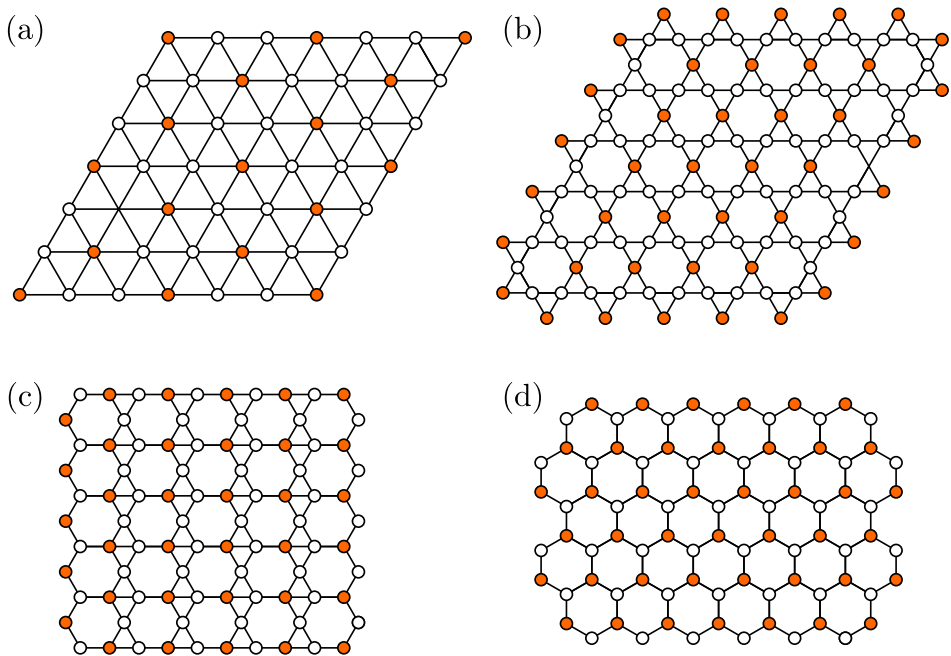}
	\caption{\label{fig:fig_7} Four lattices for which we identify the maximum independent set $\alpha(G)$ allowing us to determine the scaling of the gap between upper and lower bound for entanglement for general $N$. The maximum independent set corresponds to orange vertices. The gap is finite for the triangular lattice (a), the kagome lattice (b) and the hexa-triangular lattice (c). For the hexagonal lattice (d) it is trivially zero since this lattice is bipartite.}
\end{figure}

The lattices that we consider are pictured in Fig.~\ref{fig:fig_7}. They are the triangular, kagome, hexa-triangular and hexagonal lattices. The triangular, kagome and hexagonal lattices have been shown to be universal resources for measurement-based quantum computation \cite{Nest:2006}. We are interested in the scaling of the gap between upper and lower bounds $\Delta=|\beta(G)|-|M_{max}|$. For hexagonal lattice this gap is trivial since the lattice is bipartite and therefore $\Delta_{hexagonal}=0$. For the other three lattices that gap is
\begin{eqnarray*}
	\Delta_{tri} & = & \left\lceil\frac{N}{2}\right\rceil-\sqrt{N}-2\sum_{j=1}^{\left\lfloor\frac{N-1}{3}\right\rfloor}(\sqrt{N}-3j), \\
	\Delta_{hex-tri} & = & \frac{1}{18}(12N-3\sqrt{9+12N}+9)-\left\lfloor\frac{N}{2}\right\rfloor, \\
	\Delta_{kag} & = & \frac{1}{9}(6N-\sqrt{13+3N}-11)-\left\lfloor\frac{N}{2}\right\rfloor,
\end{eqnarray*}
where $\Delta_{tri}$ is valid for $L>3$, with $L^2=N$ being the number of vertices on one side of the triangular lattice. The particular form of $\Delta$ depends on the boundary conditions. However the general behavior of the scaling remains unchanged for different boundary conditions. Interestingly the lower bound for all four lattices is the same, $|M_{max}|=\left\lfloor\frac{N}{2}\right\rfloor$. The scaling of $\Delta$ is pictured in Fig.~\ref{fig:fig_8}.

\begin{figure}[t]
	\centering
	\includegraphics[scale=0.9]{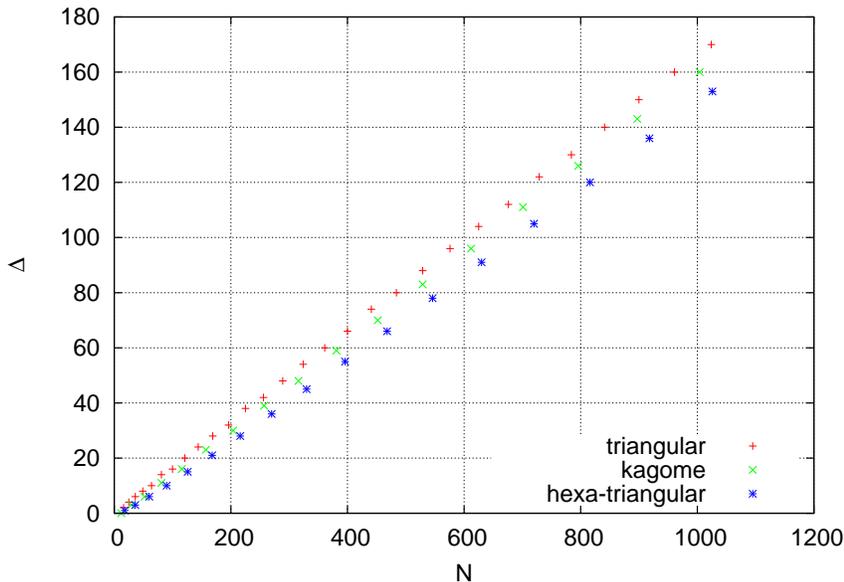}
	\caption{\label{fig:fig_8} Scaling of the gap $\Delta$ for the three lattices pictured in Fig.~\ref{fig:fig_6}(a)-(c). The gap $\Delta$ increases linearly in the leading term.}
\end{figure}

\section{\label{sec:Alternative}Alternative description of CSS}

In this section we focus only on one measure, the relative entropy of entanglement, and we abandon the stabiliser generator description of graph states to see if it is possible to arrive at the closest separable state $\omega$ using complementary methods. We consider two approaches. The first one is inspired by projected entangled pairs states description of graph states. Usually PEPs methods are used to describe pure entangled states, however we adapt this approach to construct the closest separable state $\omega$. The second approach of obtaining $\omega$ relies on introducing optimal amounts of noise in the form of relative phases and averaging over these phases. The success of both of these methods rests on our ability to identify the maximum independent set $\alpha(G)$.

\subsection{\label{subsec:PEPS}PEPs construction}

We briefly highlight the PEPs description of entangled states \cite{Verstraete:2004}. Consider a graph state $|G\rangle$ with $N$ qubits. A PEPs $|\Psi\rangle\in\mathbb{C}^{2^{N}}$ is constructed by replacing a physical qubit $a$ by $|N_a|$ virtual qubits where $|N_a|$ denotes the degree of vertex $a$. Each physical edge $(a,b)$ is then replaced by maximally entangled state of the corresponding two virtual qubits $|G_2\rangle=\frac{1}{\sqrt{2}}(|0+\rangle+|1-\rangle)$. The original graph state $|G\rangle$ can then be obtained by applying a projector $P_a:=|0\rangle\langle0|^{a^{1}}\ldots\langle0|^{a^{|N_{a}|}}+|1\rangle\langle1|^{a^{1}}\ldots\langle1|^{a^{|N_{a}|}}$ at each physical site.

The above approach can be adapted to describe separable mixed states with few changes. Instead of maximally entangled pairs of virtual qubits the basic building blocks are maximally correlated separable pairs of qubits $\omega_2$. Maximally correlated in this sense means that the relative entropy between the separable states and the tensor product of its subsystems is unity, $S(\omega_2||\omega_2^{(1)}\otimes\omega_2^{(2)})=S(\omega_2||\frac{1}{4}I\otimes I)=1$ where $\omega_2^{(1)}=\textrm{Tr}_2[\omega_2]$ and similarly for $\omega_2^{(2)}$. In fact it will be necessary to use two various separable states $\omega_2$. The projectors being applied to physical sites will also have a different structure compared to the case of pure entangled states.

Define two 2-qubit maximally correlated separable states of virtual qubits $i'$ and $j'$
\begin{equation}\label{eq:two_omegas}
	\eqalign{\omega_{i'j'}^A:=|+0\rangle\langle+0|+|-1\rangle\langle-1| \cr
	\omega_{i'j'}^B:=|0+\rangle\langle0+|+|1-\rangle\langle1-|.}
\end{equation}
Here and in the rest of this subsection we omit the normalisation constants. Virtual qubits are denoted by primed letters $a'$ and physical sites by $a$. These states are in fact both closest separable states to a 2-qubit graph state. It is also useful to give these states a graphical representation depicted in Fig.~\ref{fig:fig_9}.
\begin{figure}[t]
	\centering
	\includegraphics[scale=2.3]{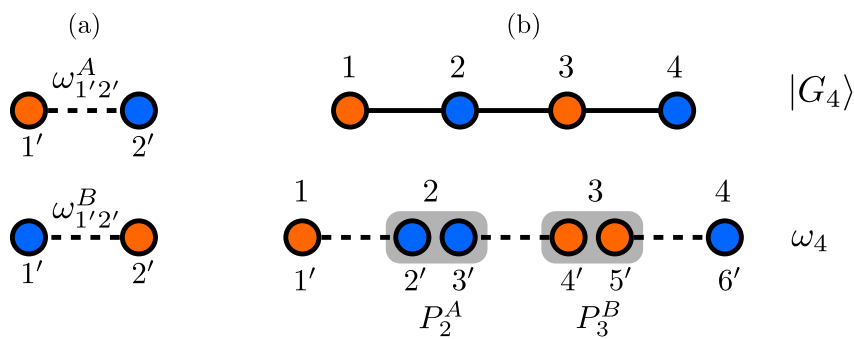}
	\caption{\label{fig:fig_9} (a) The two 2-qubit states from Eq.~\ref{eq:two_omegas} used in construction of closest separable states. The dashed edges represent that the states are separable. We choose the convention that orange vertex corresponds to a qubit in $\{|\pm\rangle\}$ basis and blue vertex corresponds to a qubit in $\{|0\rangle,|1\rangle\}$ basis. (b) Open linear 4-qubit graph state $|G_4\rangle$ and its corresponding closest separable state $\omega_4$.}
\end{figure}
Using the separable states in Eq.~(\ref{eq:two_omegas}) we construct a separable state of virtual qubits by placing either $\omega_{i'j'}^A$ or $\omega_{i'j'}^B$ on edges of the graph state. The 2-qubit separable states are picked in such a way that all virtual qubits at a physical site $a\in\alpha(G)$ are orange. In the case of bipartite graph states this means that all virtual qubits at physical sites $b\in\beta(G)$ will be of the same colour, blue. However this is not true anymore in the case of non-bipartite graph states where virtual qubits at a physical site $b\in\beta(G)$ will be of both colours. This is illustrated in Fig.~\ref{fig:fig_8}. Finally in analogy to the usual PEPs construction the virtual qubits at physical sites are projected using the following maps
\begin{eqnarray*}\label{eq:projectors_bipartite}
	P_a^A & := & |0\rangle\langle0\ldots0|+|1\rangle\langle1\ldots1|,\nonumber \\
	P_a^B & := & |+\rangle\langle\tilde{+}|+|-\rangle\langle\tilde{-}|.
\end{eqnarray*}
The projector $P_a^A$ is applied if the physical site $a\in\beta(G)$ contains at least one virtual qubit of blue color and $P_a^B$ is applied if $a\in\alpha(G)$ which means that all its virtual qubits are orange. $|\tilde{+}\rangle$ is an equal superposition of all tensor product states $\{|+\rangle,|-\rangle\}$ which are $+1$ eigenstates of $X\otimes\ldots\otimes X$. Similarly $|\tilde{-}\rangle$ is an equal superposition of all $-1$ eigenstates of $X\otimes\ldots\otimes X$.

How this construction works is most easily seen when illustrated on an explicit example. Consider an open linear graph state of 4 qubits as in Fig.~\ref{fig:fig_8}. Following the 2-colouring of the graph the tensor product of 6 virtual qubits takes the following form:
\begin{equation*}
	\omega_6'=\omega_{1'2'}^A\otimes\omega_{3'4'}^B\otimes\omega_{5'6'}^A.
\end{equation*}
The desired 4-qubit closest separable state is then obtained by applying the projector
\begin{equation*}
	\omega_4=(P_2^A\otimes P_3^B)\omega_6'(P_2^A\otimes P_3^B),
\end{equation*}
where $P_2^A=|0\rangle\langle00|+|1\rangle\langle11|$ acts on physical site 2 and $P_3^B=|+\rangle\langle++|+|+\rangle\langle--|+|-\rangle\langle+-|+|-\rangle\langle-+|$ acts on physical site 3.

\subsection{\label{subsec:Noise} Noise construction}

In this subsection we focus on the third approach of constructing the closest separable state $\omega$. The common feature of both previous approaches is the necessity of identifying either the maximum independent set $\alpha(G)$ or the minimum vertex cover $\beta(G)$. This remains true for the approach presented here as well. However this time we ask the question if there is a simple way of obtaining $\omega$ by introducing noise to the pure state rather than resorting to methods based on stabiliser generators or pairs of maximally correlated separable states.

This new approach relies on introducing distinct relative phases to certain qubits and then averaging over them to obtain $\omega$. The minimum number of the relative phases is equal to the cardinality of the minimum vertex cover $|\beta(G)|$.

The graph state vector can be written in the following form \cite{Briegel:2001}
\begin{equation}\label{eq:graph_vector}
	|G\rangle=\frac{1}{2^{N/2}}\bigotimes_{j=1}^{N}(|0\rangle_j+|1\rangle_j Z_{N'_j}),
\end{equation}
where the Pauli $Z$ matrix is applied to a subset of qubit $j$'s neighborhood $N'_j:=\{i\in N_j|i>j\}$. Now lets define a new vector $|\Phi\rangle$ that differs from the graph state in Eq.~(\ref{eq:graph_vector}) in that it contains the above mentioned relative phases $\phi_j$
\begin{equation*}
	|\Phi\rangle:=\frac{1}{2^{N/2}}\bigotimes_{j=1}^N(|0\rangle_j+e^{i\phi_jm(j)}|1\rangle_j Z_{N'_j}),
\end{equation*}
where $m(j):V\rightarrow\{0,1\}$ is an indicator function from the set of all vertices $V$ given by
\begin{equation*}
	m(j) := \left\{
		\begin{array}{cl}
			0 & \textrm{if } j\in\alpha(G),\\
			1 & \textrm{if } j\in\beta(G).
		\end{array} \right.
\end{equation*}
So the relative phase $\phi_j$ is applied only to qubits that correspond to vertices in the minimum vertex cover. For example a 3-qubit open linear graph state with a relative phase is $|\Phi\rangle=\frac{1}{2\sqrt{2}}(|0\rangle_1+|1\rangle_1Z_2)\otimes(|0\rangle_2+e^{i\phi_{2}}|1\rangle_2Z_3)\otimes(|0\rangle_3+|1\rangle_3)$. Finally the closest separable state is given by averaging over these phases
\begin{equation*}
	\omega=\frac{1}{(2\pi)^{|\beta(G)|}}\int_{0}^{2\pi}d\phi|\Phi\rangle\langle\Phi|,
\end{equation*}
where $d\phi=\Pi_{\{j|j\in\beta(G)\}}d\phi_j$.

\section{\label{sec:Conclusions}Conclusions}

We have presented a method of evaluating three multipartite entanglement measures in pure graph states. Our approach uses simple group theoretic arguments to identify a suitable subspace of the original Hilbert space whose properties can be used to find the relevant closest separable and product states as well as the minimal linear decomposition of a pure graph state $|G\rangle$.

The problem of evaluating entanglement measures can be mapped directly to a well known problem of identifying the maximum independent set in graph theory. Knowing the size of the maximum independent set corresponds to knowing the upper bound for all three entanglement measures as well as whether the upper and lower bounds are equal. Identifying which qubits comprise the maximum independent set allows us to construct the minimal linear decomposition of $|G\rangle$ into product states as well as its closest separable and closest product state. 

The closest separable state $\omega$ admits a non-stabiliser description using a PEPs inspired construction. This immediately begs the question whether a suitable $\omega$ can be constructed for weighted graph states. Any realistic scheme of preparing graph states will use entangling gates best described by a general control phase gate where the phase does not always have the ideal value of $\phi=\pi$. Rather it is picked randomly from some distribution centered around this ideal value.

Our techniques developed so far are the first step towards investigating total entanglement properties of such realistic systems. Furthermore our methods can be used in the study of entanglement in lattice models with long-range interactions. The strength and range of these Ising-type interactions can be captured by the phase in the entangling gate. All work on this topic has so far been limited to the study of bipartite entanglement measures. We will present our findings on this topic in a separate paper.

Interesting observation is that the lower bound for entanglement of any pure graph state can be found efficiently since the maximum matching problem of an arbitrary graph is in the P complexity class. On the other hand the maximum independent set problem, and hence also evaluation of the upper bound, is NP-hard for a general graph. One exception are bipartite graphs for which $|M_{max}|=N-|\alpha(G)|$. Efficient estimation of the upper bound is still an open problem. It would be interesting to see whether separable states that approximate the closest separable state can be constructed efficiently.

A closely related question is whether there is some deeper relationship between the upper and the lower bound. Calculations for graph states up to 10 qubits suggest that the gap between the bounds grows very slowly. It would be interesting to see how this changes for the case of larger and more general graph states. We have demonstrated that for certain regular lattices in two spatial dimensions the gap $\Delta$ increases linearly with the number of qubits $N$. This is not too surprising. Because we have considered regular lattices, linear increase of $N$ results in a linear increase of the cardinality of the maximum independent set $|\alpha(G)|$ while the size of the maximum matching $|M_{max}|$ remains constant. However this behaviour is not likely to be true for more general graphs. Particularly interesting would be to find a relationship between the size of the gap and some structural quantities of the underlying graph that can be computed efficiently.

This naturally leads to the final question which is concerned with evaluating entanglement in graph states where the bounds are not equal. In this case our methods can achieve the upper bound and therefore do not say anything concrete about the actual entanglement of the graph state. Numerical evidence suggests that for certain states geometric measure is less than the upper bound \cite{Chen:2010} whereas for some other states it is equal to the upper bound. An open question is to see if the three entanglement measures are still equal when the bounds are different and to characterize states whose upper bound is the actual value for entanglement.

\ack
MH and MM acknowledge financial support from Japan Society for the Promotion of Science (JSPS) by KAKENHI grants No. 23-01770, No. 23540463 and No. 2324000.

\section*{References}
\bibliographystyle{unsrt}
\bibliography{direct-bib}

\end{document}